\documentclass[a4paper,11pt]{article}

\input{preamble.tex}

\begin{document}

\title{Popular Matching in Roommates Setting is NP-hard}

\author{
Sushmita Gupta\thanks{University of Bergen, Bergen, Norway. \texttt{sushmita.gupta@ii.uib.no}}
 \and Pranabendu Misra\thanks{University of Bergen, Bergen, Norway. \texttt{pranabendu.misra@ii.uib.no}}
 \and  Saket Saurabh\thanks{The Institute of Mathematical Sciences, HBNI, Chennai, India. \texttt{saket@imsc.res.in}}
 \and Meirav Zehavi\thanks{Ben-Gurion University, Beersheba, Israel. \texttt{meiravze@bgu.ac.il}} 
}

\maketitle

\thispagestyle{empty}

\begin{abstract} 

An input to the \popular\ problem, in the roommates setting, consists of a graph $G$ and each vertex ranks its neighbors in strict order, known as its preference. In the  \popular\ problem the objective is to {\em test} whether there exists a matching $M^\star$ such that there is no matching $M$ where more people are happier with $M$ than with $M^\star$.  In this paper we settle the computational complexity of the \popular\ problem in the roommates setting by showing that the problem is \NPC. Thus, we resolve 
an open question that has been repeatedly, explicitly asked over the last decade. 



\end{abstract}

\newpage
\pagestyle{plain}
\setcounter{page}{1}


\newcommand{\Ma}[1]{\textcolor{magenta}{#1}}
\newcommand{\ma}[1]{\todo[line, color={Yellow}]{\footnotesize{#1}}\xspace}
\newcommand{\il}[1]{\todo[inline, color={Yellow}]{{#1}}}
\newcommand{\etalcite}[1]{\emph{et~al.}~\cite{#1}}
\newcommand{\etal}{\emph{et~al.}\xspace}

\section{Introduction}\label{sec:intro}
Matching problems with preferences are ubiquitous in everyday life scenarios. They arise in applications such as the assignment of students to universities, doctors to hospitals, students to campus housing, pairing up police officers, kidney donor-recipient pairs and so on. The common theme is that individuals have preferences over the possible outcomes and the task is to find a matching of the participants that is in some sense optimal with respect to these preferences. In this paper we study the computational complexity of computing one such solution concept, namely the  \popular ~problem. The input to the \popular\ problem consists of a graph on $n$ vertices and the preferences of the vertices represented as a ranked list of the neighbors of every vertex, said to be the {\it preference list} of the vertex. The goal is to find a {\it popular matching}--a matching that is preferred over any other matching (in terms of the preference lists) by at least half of the vertices in the graph.  

Popular matching finds real-life applications in avenues as diverse as the organ-donor exchange markets, spectrum sharing in cellular networks, barter exchanges, to just name a few \cite{Manlove13b,XiaoHanYuenDaSilva16}. Specifically, situations in which a {\it stable matching} -- a matching that does not admit a \emph{blocking edge}, i.e. an edge whose endpoints prefer each other to their respective ``situation'' in the current matching -- is too restrictive, popular matching finds applicability. It is known that stable matching is the smallest sized popular matching. So for applications where it is desirable to have matchings of larger size than a stable matching -- for instance, allocating projects to students, or pairing up police officers, where the absence of blocking edges is not mandatory -- popular matching may be a suitable alternative. The notion of popularity captures a natural relaxation of the notion of stability: blocking edges are permitted but the matching, nevertheless has overall stability

To define the \popular\ problem formally, we first need few  definitions. 
Let $N_G(v)$ denotes the neighborhood of a vertex $v\in V(G)$.  Given a vertex $v\in V(G)$, a {\em preference list of $v$ in $G$} is a {\em bijective} function $\ell_v: N_G(v)\rightarrow \{1,2,\ldots,|N_G(v)|\}$. Informally, the smaller the number a vertex $v\in V(G)$ assigns to a vertex $u\in N_G(v)$, the more $v$ prefers to be matched to $u$. In particular, for all $u,w\in N_G(v)$, if $\ell_v(u)<\ell_v(w)$, then $v$ prefers $u$ over $w$. A matching $M$ in $G$ is a subset of $E(G)$ whose edges are pairwise disjoint. We say that a vertex $v\in V(G)$ is matched by a matching $M$ if there exists a (unique) vertex $u\in V(G)$ such that $\{u,v\}\in M$, which we denote by $u=M(v)$. 

In literature, the terminology related to \popular\ is closely related to that of the {\sc Stable Marriage} problem. 
When the input graph is (bipartite) arbitrary, the instance is said to be that of the {\em (stable marriage)  roommates setting} of the problem.  
We denote an instance of \popular\ (in the roommates setting) by $I=(G,L=\{\ell_v: v\in V(G)\})$. Roughly speaking, a vertex $v\in V(G)$ prefers a matching $M$ over a matching $M'$ if its ``status'' in $M$ is better than the one in $M'$, where being not matched is the least preferred status. Formally, the notion of preference over matchings is defined as follows.
Given two matchings in $G$, denoted by $M$ and $M'$, we say that a vertex $v\in V(G)$ {\em prefers} $M$ over $M'$ if one of the following conditions is satisfied: {\em (i)} $v$ in matched by $M$ but not matched by $M'$; {\em (ii)} $v$ is matched by both $M$ and $M'$, and $\ell_v(M(v))<\ell_v(M'(v))$.  We say that $M'$  is {\em more popular} than $M$, if the number of vertices that prefer $M'$ 
to $M$ exceeds the number of vertices that prefer $M$ to $M'$. A matching $M$ is {\em popular} if and only if there is no matching $M'$ that is more popular than $M$. 
In the decision version of the \popular\ problem, given an instance  $I=(G,L=\{\ell_v: v\in V(G)\})$, 
{\em the question is whether  there exists a popular matching?}

%

\medskip



\noindent 
{\bf History of the problem and our result.} 
The provenance of the notion of a popular matching can be dated to the work of Condorcet in 1785 on the subject of a {\it Condorcet winner} \cite{Condorcet1785}. In the last century, however, the notion was introduced as the {\it majority assignment} by G\"{a}rdenfors \cite{Gardenfors75} in 1975. Anecdotal retelling ascribes the coinage of the term ``popular matching'' and the associated question of does there exist a polynomial time algorithm for \popular\, in the ``housing allocation'' setting,  to Robert Irving during a talk in University of Glasgow in 2002, \cite[pg 333]{Manlove13b}. In 2005, Abraham \etalcite{Abraham07} was the first to discuss an efficient algorithm for computing a popular matching albeit for the case where the graph is bipartite and only the vertices in one of the partitions have a preference list, a setting known as the {\em housing allocation}. The persuasive motivation and elegant analysis of Abraham \etal led to a spate of papers on popular matching \cite{ManloveSng06,HuangKavithaMichailNasre08,KavithaN09,KavithaMestreNasre11,OptimalPM09,BiroIrvingManlove10,HuangKavitha17,KiralyKarkus17} covering diverse settings that include strict preferences as well as one with ties. It is well-known that when the input graph is bipartite--the stable marriage setting, \popular\ problem can always be decided affirmatively in polynomial time, \cite{DBLP:journals/iandc/HuangK13}. 
It is equally well-known that when the graph is arbitrary, the computational complexity of deciding whether a popular matching exists is unknown. In particular, whether \popular\ is \NPH\ has been repeatedly, explicitly asked as an open question over the last decade 
~\cite{EgresWiki,BiroIrvingManlove10,CsehTalk2015,Cseh-survey,CsehKavitha16,HuangKavitha13j,DBLP:journals/iandc/HuangK13,HuangKavitha17,Size-Popularity-Tradeoff-Kavitha14j,DBLP:journals/corr/abs-1802-07440,KiralyKarkus17,ManloveSummerSchool2013,Manlove13b}. 
Indeed, it has been stated as one of the main open problem in the area (see the aforementioned citations). In this paper we settle this question by proving the following result. 
\begin{theorem}\label{thm:main}
\popular\ is \NPC.
\end{theorem}

\smallskip

\noindent 
{\bf Our method.}  An optimization question related to the \popular\ is about finding a popular matching of the largest size (as not all popular matchings are of same size). Let this problem be called {\sc Max Sized Popular Matching}. Until recently, it was also not known  whether this problem is \NPH\ in roommate setting. 
Recently, Kavitha investigated the computational complexity of {\sc Max Sized Popular Matching} in arbitrary graphs \cite{DBLP:journals/corr/abs-1802-07440}, and showed it to be \NPH.  This reduction serves as one of the main three gadgets 
in our reduction--the other two gadgets are completely new. The design of our reduction required several new insights. 
Firstly, our source problem is a ``{\sc 3-SAT}-like'' variant of {\sc Vertex Cover}, which allows us to enjoy benefits of both worlds: we gain both the lack of ``optimization constraints'' as in {\sc 3-SAT}, and the simplicity of {\sc Vertex Cover}. The usage of this source problem requires us to encode selection of {\em exactly} one ``element'' out of two, and exactly two ``elements'' out of three. Here, our gadget design is carefully tailored to exploit a known characterization of a popular matching. In particular,  
we make use of ``troublemaker triangles''--these are triangles consisting of three vertices, one of whom must be matched to a vertex  outside the triangle to give rise to a popular matching. We embed these triangles in a structure that coordinates the way in which they can be {\em traversed}.  Here, traversal precisely refers to the above mentioned characterization, which relies on the exposure of certain alternating paths and cycles in a graph associated with a candidate matching (to be a popular matching). Our gadgets lead traversals of such paths and cycles to dead-ends. We remark that when we describe our gadgets, we present additional intuitive explanations of their design.



\medskip

\noindent
{\bf Related results.} Chung~\cite{Chung00} was the first to study the \popular\ problem in the roommates setting. He observed that every stable matching is a popular matching. In the midst of a long series of articles, the issue of the computational complexity of \popular\ in an arbitrary graph remained unsettled, leading various researcher to devise notions such as the {\it unpopularity factor} and {\it unpopularity margin} \cite{HuangKavithaMichailNasre08,McCutchen08,HuangKavitha13j} in the hope of capturing the essence of popular matchings. A solution concept that emerged from this search is the {\it maximum sized popular matching}, motivated by the fact that unlike stable matchings (Rural Hospital Theorem \cite{Roth86}), all popular matchings in an instance do not match the same set of vertices or even have the same size. Thus, it is natural to focus on the size of a popular matching.  
There is a series of papers that focus on the {\sc Max Sized Popular Matching} problem in bipartite graphs (without ties in preference lists) \cite{Size-Popularity-Tradeoff-Kavitha14j,DBLP:journals/iandc/HuangK13,CsehKavitha16} and (with ties) \cite{PM-2sidedPref-1sidedTies-CsehHuangKavitha17j}. When preferences are strict, there are various polynomial time algorithm that solve {\sc Max Sized Popular Matching} in bipartite graphs: Huang and Kavitha \cite{HuangKavitha13j} give an $\mathcal{O}(mn_{0})$ algorithm that is improved by Kavitha to $\mathcal{O}(m)$ \cite{Size-Popularity-Tradeoff-Kavitha14j} where $m$ and $n_{0}$ denote the number of edges in the bipartite graph and the size of the smaller vertex partition, respectively. In the presence of ties (even on one side), the {\sc Max Sized Popular Matching}  was shown to be \NPH\ \cite{PM-2sidedPref-1sidedTies-CsehHuangKavitha17j}. It is worth noting that every stable matching is popular, but the converse is not true. As a consequence of the former, every bipartite graph has a popular matching that is computable in polynomial-time because it has a stable matching computable by the famous Gale-Shapley algorithm described in the seminal paper \cite{GaleShapley62} by the eponymous Gale and Shapley. 

\section{Preliminaries}\label{sec:prelims}

\paragraph{Standard Definitions and Our Notation.} Given a graph $G$, we let $V(G)$ and $E(G)$ denote the vertex set and edge set of $G$, respectively. Throughout the paper, we consider undirected simple graphs. We view an edge as a {\em set} of two vertices. A {\em triangle} in $G$ is a cycle in $G$ on exactly three vertices. The neighborhood of a vertex $v\in V(G)$ in $G$ is denoted by $N_G(v)=\{u\in V(G): \{u,v\}\in E(G)\}$, and the set of edges incident to $v$ in $G$ is denoted by $E_G(v)$. Given a vertex $v\in V(G)$, a {\em preference list of $v$ in $G$} is a {\em bijective} function $\ell_v: N_G(v)\rightarrow \{1,2,\ldots,|N_G(v)|\}$. Informally, the smaller the number a vertex $v\in V(G)$ assigns to a vertex $u\in N_G(v)$, the more $v$ prefers to be matched to $u$. In particular, for all $u,w\in N_G(v)$, if $\ell_v(u)<\ell_v(w)$, then $v$ prefers $u$ over $w$. A matching $M$ in $G$ is a subset of $E(G)$ whose edges are pairwise disjoint. We say that a vertex $v\in V(G)$ is matched by a matching $M$ if there exists a (unique) vertex $u\in V(G)$ such that $\{u,v\}\in M$, which we denote by $u=M(v)$. Moreover, $M$ is maximal if there is no edge in $E(G)$ such that both endpoints of that edge are not matched by $M$.

We denote an instance of \popular\ (in the roommates setting) by $I=(G,L=\{\ell_v: v\in V(G)\})$. Roughly speaking, a vertex $v\in V(G)$ prefers a matching $M$ over a matching $M'$ if its ``status'' in $M$ is better than the one in $M'$, where being not matched is the least preferred status. Formally, the notion of preference over matchings is defined as follows.

\begin{definition}\label{def:preference}
Let $I=(G,L=\{\ell_v: v\in V(G)\})$ be an instance of \popular. Given two matchings in $G$, denoted by $M$ and $M'$, we say that a vertex $v\in V(G)$ {\em prefers} $M$ over $M'$ if one of the following conditions is satisfied: {\em (i)} $v$ in matched by $M$ but not matched by $M'$; {\em (ii)} $v$ is matched by both $M$ and $M'$, and $\ell_v(M(v))<\ell_v(M'(v))$. The number of vertices in $V(G)$ that prefer $M$ over $M'$ is denoted by $\vote(M,M')$.
\end{definition}

Roughly speaking, $\vote(M,M')$ above can be thought of as the number of vertices that will vote in favor of $M$ when they are asked to decide whether $M$ or $M'$ should be chosen. For notational convenience, given a vertex $v\in V(G)$, we denote $\ell_v(v)=|N_G(v)|+1$, and given a matching $M$ where $v$ is not matched, we denote $M(v)=v$. Then, for example, the first condition in Definition \ref{def:preference} is subsumed by the second one. We now also formally define the notion of popularity.

\begin{definition}\label{def:popularity}
Let $I=(G,L=\{\ell_v: v\in V(G)\})$ be an instance of \popular. We say that a matching $M$ in $G$ is {\em popular} if $\vote(M',M)-\vote(M,M')\leq 0$ for any other matching $M'$ in $G$.
\end{definition}

Intuitively, the meaning of the definition above is that when the vertices are asked whether we should replace $M$ by $M'$, for any other matching $M'$, the number of vertices that will vote against the swap is at least as large as the number of vertices that will vote in favor of it. Let us recall that in the \popular\ problem, the objective is to decide whether there exists a popular matching.

Given a graph $G$, we say that a vertex $v\in V(G)$ {\em covers} an edge $e\in E(G)$ if $v$ is incident to $e$, that is, $v\in e$. A {\em vertex cover} $U$ in $G$ is a subset of $V(G)$ such that every edge in $E(G)$ is covered by at least one vertex in $U$. In the {\sc Vertex Cover} problem, we are given a graph $G$ and an integer $k$, and the objective is to decide whether $G$ has a vertex cover of size at most $k$.

\paragraph{Known characterization of popular matchings.} We need to present (known) definitions of a labeling of the edges in $E(G)$ as well as of a special graph derived from $G$ and a matching $M$ in $G$, which will give rise to a characterization of popular matchings.

\begin{definition}[Definition 2 in \cite{DBLP:journals/iandc/HuangK13}, Rephrased]\label{def:label}
Let $I=(G,L=\{\ell_v: v\in V(G)\})$ be an instance of \popular. Given a matching $M$ in $G$, the edge labeling $\lab_M: (E(G)\setminus M)\rightarrow \{-2,0,+2\}$ is defined as follows.
$$
\lab_M(\{u,v\})=
\begin{cases}
-2\ \ \ \mathrm{if}\ \ell_u(M(u))<\ell_u(v)\ \mbox{and }\ \ell_v(M(v))<\ell_v(u)\\
+2\ \ \ \mathrm{if }\ \ell_u(M(u))>\ell_u(v)\ \mbox{and }\ \ell_v(M(v))>\ell_v(u)\\
\ \ 0\ \ \ \mathrm{otherwise} 
\end{cases}
$$
\end{definition}

Intuitively, an edge in the definition above is assigned $-2$ if both its endpoints do not prefer being matched to each other over their status in $M$, and it is assigned $+2$ if both its endpoints prefer being matched to each other over their status in $M$. 

\begin{definition}[\cite{DBLP:journals/iandc/HuangK13}]\label{def:GM}
Let $I=(G,L=\{\ell_v: v\in V(G)\})$ be an instance of \popular. Given a matching $M$ in $G$, the graph $G_M$ is the subgraph of $G$ with $V(G_M)=V(G)$ and $E(G_M)=\{\{u,v\}\in E(G): \{u,v\}\in M$ or $\lab_M(\{u,v\})\neq -2\}$. 
\end{definition}

Before we can present the characterization, we need to define the notions of an alternating path and an alternating cycle in $G_M$. First, an {\em alternating cycle} in $G_M$ is a cycle in $G_M$ (with an even number of edges) such that if we traverse the edges of the cycle (in any direction), then every edge in $M$ is followed by an edge outside $M$, and every edge outside $M$ is followed by an edge in $M$. Similarly, an {\em alternating path} in $G_M$ is a path in $G_M$ such that if we traverse the edges of the path (in any direction), then every edge in $M$ is followed by an edge outside $M$ (with the exception of the last edge), and every edge outside $M$ is followed by an edge in $M$ (with the same exception), and in addition, if the edge incident to the first or last vertex on the path is not in $M$, then that vertex is not matched by $M$. Now, the characterization is given by the following proposition.

\begin{proposition}[Theorem 1 in \cite{DBLP:journals/iandc/HuangK13}, Rephrased]\label{prop:char}
Let $I=(G,L=\{\ell_v: v\in V(G)\})$ be an instance of \popular. A matching $M$ in $G$ is popular if and only if the following conditions hold in $G_M$.
\begin{itemize}
\item There is no alternating cycle in $G_M$ that contains at least one edge labeled +2 by $\lab_M$.
\item There is no alternating path in $G_M$ that starts from a vertex not matched by $M$ and contains at least one edge labeled +2 by $\lab_M$.
\item There is no alternating path in $G_M$ that contains at least two edges labeled +2 by $\lab_M$.
\end{itemize}
\end{proposition}

We remark that the observation that Theorem 1 in \cite{DBLP:journals/iandc/HuangK13} holds for general graphs (its statement refers to bipartite graphs) is noted on page 6 of that paper. The usefulness of Proposition \ref{prop:char} for us is that it will help us verify that the matching we construct when we prove that forward direction of the correctness of our reduction is indeed popular. Note that, if we are to prove the popularity of matching by using only the definition of popularity, then we need to compare the matching to a huge number of other matchings (that can be of a super-exponential magnitude). Thus, Proposition \ref{prop:char} will come in handy.

\section{Definition of \pvc}\label{sec:partVC}

The correctness of our reduction will crucially rely on the fact that our source problem will {\em not} be {\sc Vertex Cover}, but a variant of it that we call \pvc. This variant is defined as follows.

\paragraph{Problem definition.} The input of \pvc\ consists of a graph $G$, a collection ${\cal P}$ of pairwise disjoint edges in $G$, and a collection ${\cal T}$ of pairwise disjoint sets of size 3 of vertices that induce triangles in $G$,\footnote{That is, for all $\{x,y,z\}\in{\cal T}$, we have that $\{x,y\},\{y,z\},\{z,x\}\in E(G)$.} such that every vertex in $V(G)$ occurs in either a triangle in ${\cal T}$ or an edge in ${\cal P}$ (but not in both). In other words, ${\cal T}\cup{\cal P}$ forms a partition of $V(G)$ into sets of sizes 3 and 2.

To ease readability, we will refer to a set (edge) in $\cal P$ as a {\em pair} and to a set in $\cal T$ as a {\em triple}. 

The objective of \pvc\ is to decide whether $G$ has a vertex cover $U$ such that the two following conditions hold.
\begin{enumerate}
\item\label{pvc:condition1} For every $P\in {\cal P}$, it holds that $|U\cap P|=1$.
\item\label{pvc:condition2} For every $T\in {\cal T}$, it holds that $|U\cap T|=2$.
\end{enumerate}
A vertex cover $U$ with the properties above will be referred to as a {\em solution}.

\paragraph{Remark and Hardness.} We remark that it will be crucial for us that {\em (i)} the sets in ${\cal T}\cup{\cal P}$ are all pairwise disjoint, {\em (ii)} the maximum size of a set in ${\cal T}\cup{\cal P}$ is only 3 and all but one of the vertices of a set in ${\cal T}\cup{\cal P}$ must be selected, and {\em (iii)} all solutions must have the same size, where the implicit size requirement (that is, being of size exactly $|{\cal P}|+2|{\cal T}|$) is automatically satisfied if Conditions \ref{pvc:condition1} and \ref{pvc:condition2} are satisfied.

Now, we claim that \pvc\ is \NPH. The correctness of this claim directly follows from a classic reduction from {\sc 3-SAT} to {\sc Vertex Cover} (see, e.g., \cite{sipser2006}). For the sake of completeness, we present this reduction and argue formally why its output can be viewed correctly as an instance of \pvc\ (rather than an instance of {\sc Vertex Cover}) in Appendix \ref{app:partVC}.

\begin{lemma}\label{lem:partVC}
\pvc\ is \NPH.
\end{lemma}
\section{Reducing \pvc\ to \popular}\label{sec:reduction}

Let $I=(G,{\cal P},{\cal T})$ be an instance of \pvc. In this section, we construct an instance $\red(I)=(H,L=\{\ell_v: v\in V(H)\})$ of \popular. Note that, to avoid confusion, we denote the graph in $\red(I)$ by $H$ rather than $G$, since the latter already denotes the graph in $I$.
We remark that the Edge Coverage gadget below is in fact the entire reduction from (standard) {\sc Vertex Cover} to an optimization variant of \popular\ recently given by Kavitha \cite{DBLP:journals/corr/abs-1802-07440} (in that context, we will use notation consistent with this work). Our two other gadgets are completely new. After describing the Edge Coverage gadget, we briefly discuss its weakness. In particular, this brief discussion sheds light on the jump in understanding the {\sc Popular Matching} problem that we had to perform in order to employ this known gadget (or any other similar gadget in the literature on popular matchings) to prove the hardness of {\sc Popular Matching}.

\subsection{Edge Coverage}

For every vertex $i\in V(G)$, we add four new vertices (to $H$), denoted by $a_i,b_i,c_i$ and $d_i$. In addition, we add the edges $\{d_i,a_i\},\{a_i,b_i\},\{a_i,c_i\}$ and $\{b_i,c_i\}$ (see Fig.~\ref{fig:edgeCoverage}). Now, for every edge $e=\{i,j\}\in E(G)$, we add two vertices, $u^e_i$ and $u^e_j$, and the edges $\{u^e_i,u^e_j\},\{b_i,u^e_i\}$ and $\{b_j,u^e_j\}$.

Let us now give a partial definition of the preference lists of the vertices added so far (see Fig.~\ref{fig:edgeCoverage}). When we will add neighbors to some of these vertices, they will be appended to the end of these partial lists, and we will not change the values that we are about to define. For every vertex $i\in V(G)$, we have the following definitions.
\begin{itemize}
\item {\bf Vertex $a_i$:} $\ell_{a_i}(b_i)=1$; $\ell_{a_i}(c_i)=2$; $\ell_{a_i}(d_i)=3$.
\item {\bf Vertex $b_i$:} $\ell_{b_i}(a_i)=1$; $\ell_{b_i}$ restricted to $\{u^e_i: e\in E_G(i)\}$ is an arbitrary bijection into $\{2,3,\ldots,|E_G(i)|+1\}$;\footnote{That is, every vertex in $\{u^e_i: e\in E_G(i)\}$ is assigned a unique integer from $\{2,3,\ldots,|E_G(i)|+1\}$, and it is immaterial to us which bijection to choose to achieve this.} $\ell_{b_i}(c_i)=|N_G(i)|+2$.
\item {\bf Vertex $c_i$:} $\ell_{c_i}(a_i)=1$; $\ell_{c_i}(b_i)=2$.
\item {\bf Vertex $d_i$:} $\ell_{d_i}(a_i)=1$.
\item {\bf Vertex $u^e_i$ for any $e=\{i,j\}\in E_G(i)$:} $\ell_{u^e_i}(u^e_j)=1$; $\ell_{u^e_i}(b_i)=2$.
\end{itemize}

This completes the description of the Edge Coverage gadget.

\begin{figure}[t!]\centering
\fbox{\includegraphics[scale=0.8]{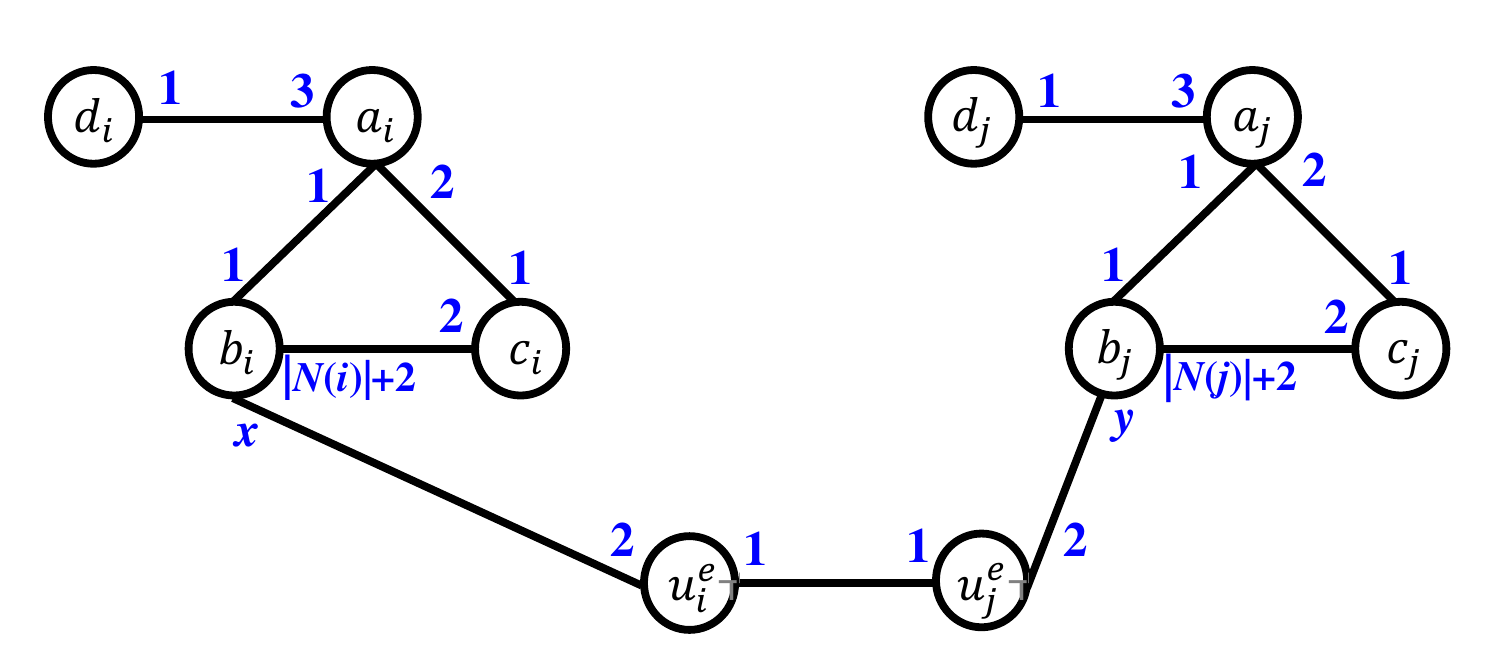}}
\caption{The Edge Coverage gadget. Here, $x\in\{2,3,\ldots,|N_G(i)|+1\}$ and $y\in\{2,3,\ldots,|N_G(j)|+1\}$.}\label{fig:edgeCoverage}
\end{figure}

\paragraph{Intuition.} The Edge Coverage gadget aims to encode (as we will see in Section \ref{sec:correctness}) the selection of a vertex as follows. In every popular matching $M$, either $\{a_i,b_i\}\in M$ or both $\{a_i,d_i\}\in M$ and $\{b_i,c_i\}\in M$. The special choice of the preferences ensure this, where the first choice indicates that $i$ is present in the vertex cover encoded by $M$, while the second choice indicates that $i$ is not present in this vertex cover. Intuitively, $a_i$ and $b_i$ prefer each other the most, but if we choose to match them, we ``leave out'' both $d_i$ and $c_i$, which gives rise to the two configurations as above. Then, the addition of $u^e_i$ and $u^e_j$, which prefer each other the most, and which are inserted in the ``middle'' of $b_i$'s and $b_j$'s lists, respectively, will ensure that that every edge is indeed covered. To establish this last claim, it will also be important that $c_i$ prefers $a_i$ over $b_i$---this will allow us to ``move'' from the configuration of having $\{a_i,d_i\},\{b_i,c_i\}\in M$ and $\{a_j,d_j\},\{b_j,c_j\}\in M$ to one where $c_i$ and $c_j$ are matched to $a_i$ and $a_j$, respectively, when we try to exhibit a matching more popular than $M$.

While this gadget, already given by Kavitha \cite{DBLP:journals/corr/abs-1802-07440}, is very useful to us, its main drawback is that it cannot enforce popular matchings to favor the selection of $\{b_i,c_i\}\in M$ and $\{a_i,d_i\}\in M$ over $\{a_i,b_i\}\in M$. In other words, this gadget does not help us, in any way, to force the encoded vertex cover to be as small as possible. (We remark that Kavitha \cite{DBLP:journals/corr/abs-1802-07440} considers a variant of {\sc Popular Matching} where the matching should be as large as possible, and hence the inherent difficulty of the problem is circumvented.) By considering \pvc\ rather than {\sc Vertex Cover}, we do not need to deal with such ``size optimization'' constraint anymore. However, we now need to handle the constraints imposed by $\cal P$ and $\cal T$. Nevertheless, these two sets are very structured as explained in Section \ref{sec:partVC} (in sharp contrast to, say, an arbitrary instance of {\sc 3-SAT}). In fact, every detail of the gadgets described next is carefully tailored to exploit the extra structural properties of \pvc\ as much as possible, as will be made clear in Section \ref{sec:correctness}.

\subsection{Pair Selector}

For every pair $\{i,j\}\in{\cal P}$ with $i<j$, we add two new vertices (to $H$), denoted by $f_{ij}$ and $f_{ji}$, along with the edges $\{d_i,f_{ij}\},\{f_{ij},c_j\},\{d_j,f_{ji}\}$ and $\{f_{ji},c_i\}$ (see Fig.~\ref{fig:pairSelector}). In addition, we insert the edges $\{c_i,d_j\}$ and $\{c_j,d_i\}$.

We update the preference lists of the vertices as follows (see Fig.~\ref{fig:pairSelector}).
\begin{itemize}
\item {\bf Vertex $c_i$:} $\ell_{c_i}(f_{ji})=3$; $\ell_{c_i}(d_j)=4$.
\item {\bf Vertex $c_j$:} $\ell_{c_j}(f_{ij})=3$; $\ell_{c_j}(d_i)=4$.
\item {\bf Vertex $d_i$:} $\ell_{d_i}(c_j)=2$; $\ell_{d_i}(f_{ij})=3$.
\item {\bf Vertex $d_j$:} $\ell_{d_j}(c_i)=2$; $\ell_{d_j}(f_{ji})=3$.
\item {\bf Vertex $f_{ij}$:} $\ell_{f_{ij}}(d_i)=1$; $\ell_{f_{ij}}(c_j)=2$.
\item {\bf Vertex $f_{ji}$:} $\ell_{f_{ji}}(d_j)=1$; $\ell_{f_{ji}}(c_i)=2$.
\end{itemize}

Note that the definition above is valid since no vertex in $V(G)$ participates in more than one pair, and hence no integer is assigned by any function $\ell_{\diamond}$ more than once. This completes the description of the Pair Selector gadget.

\begin{figure}[t!]\centering
\fbox{\includegraphics[scale=0.8]{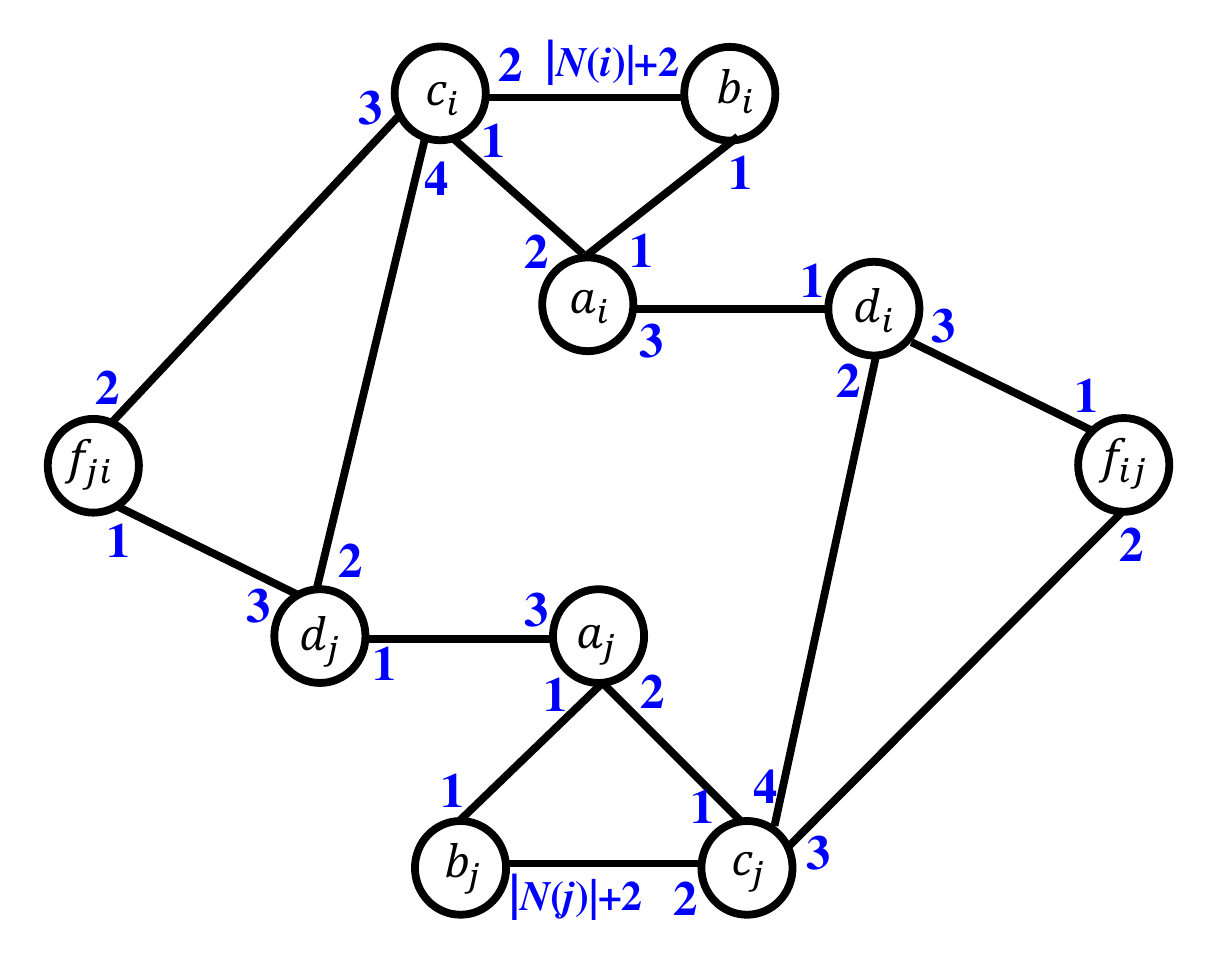}}
\caption{The Pair Selector gadget.}\label{fig:pairSelector}
\end{figure}

\paragraph{Intuition.} First, we would like to point out that the Pair Selector gadget is symmetric in the sense that if we swap $i$ and $j$, we obtain an isomorphic structure also with respect to preferences. Thus, the gadget is well (uniquely) defined even if we drop the requirement ``with $i<j$'' above. We will use this symmetry when we prove the correctness of our reduction.

To gain some deeper understanding of this gadget, let us recall that in \pvc, {\em exactly} one vertex among $\{i,j\}$ must be selected. We already know that the Edge Coverage gadget is meant to ensure that {\em at least} one vertex among $\{i,j\}$ is selected. Hence, we only need to ensure that not {\em both} $i$ and $j$ are selected. However, if both $i$ and $j$ are selected, then both $c_i$ and $d_j$ are left not matched. Then, the preferences on the triangle on $\{c_i,d_j,f_{ji}\}$ are chosen specifically to ``cause trouble''---no matter which edge of this triangle will be picked by the matching, we can replace it by a different edge on this triangle to exhibit a more popular matching. For example, if we pick $\{c_i,f_{ji}\}$, then $d_j$ is left not matched, while $f_{ji}$ prefers $d_j$ over $c_i$. This means that by replacing $\{c_i,f_{ji}\}$ by $\{f_{ji},d_j\}$, we make both $f_{ji}$ and $d_j$ more satisfied, while only $c_i$ becomes less satisfied (no other vertex in $H$ is affected by the swap).

In light of the swap above, it may appear as if it would have been sufficient to keep the triangle on $\{c_i,d_j,f_{ji}\}$, while removing the triangle on $\{c_j,d_i,f_{ij}\}$ from the gadget. However, without the second triangle, the proof of the forward direction fails---the matching attempted to construct from a vertex cover will not be popular. In particular, by having the second triangle as well, we will always by able to match all of the vertices in $H$, and hence avoid the need to consider the second condition in Proposition \ref{prop:char}. Again, we stress that the second triangle is not meant to ease the proof, but that without it the forward direction of the proof fails. It is also worth to note here that only having these two triangles is not sufficient, but the exact ``orientation'' of their preferences is crucial. In particular, if we changed the orientation of only one of the triangles---for example, if we made $c_i$ prefer $d_j$ over $f_{ji}$, $d_j$ prefer $f_{ji}$ over $c_i$, and $f_{ji}$ prefer $c_i$ over $d_j$---then the gadget would have no longer been symmetric, and the proof of the forward direction would have failed. Roughly speaking, the two triangles on $\{c_i,d_j,f_{ji}\}$ and $\{c_j,d_i,f_{ij}\}$  ``work together'' to prevent the existence of alternating cycles that must not exist by Proposition \ref{prop:char}. Deeper coordination is required in the next gadget, and we will elaborate on it more when we explain the intuition behind that gadget.

\subsection{Triple Selector}

For every triple $\{i,j,k\}\in{\cal T}$ with $i<j<k$, we add six new edges (to $H$): $\{d_i,d_j\}$, $\{d_j,d_k\}$, $\{d_k,d_i\}$, $\{c_i,c_j\}$, $\{c_j,c_k\}$ and $\{c_k,c_i\}$ (see Fig.~\ref{fig:tripleSelector}).

We update the preference lists of the vertices as follows (see Fig.~\ref{fig:tripleSelector}).
\begin{itemize}
\item {\bf Vertex $c_i$:} $\ell_{c_i}(c_k)=3$; $\ell_{c_i}(c_j)=4$.
\item {\bf Vertex $c_j$:} $\ell_{c_j}(c_i)=3$; $\ell_{c_j}(c_k)=4$.
\item {\bf Vertex $c_k$:} $\ell_{c_k}(c_j)=3$; $\ell_{c_k}(c_i)=4$.
\item {\bf Vertex $d_i$:} $\ell_{d_i}(d_j)=2$; $\ell_{d_i}(d_k)=3$.
\item {\bf Vertex $d_j$:} $\ell_{d_j}(d_k)=2$; $\ell_{d_j}(d_i)=3$.
\item {\bf Vertex $d_k$:} $\ell_{d_k}(d_i)=2$; $\ell_{d_k}(d_j)=3$.
\end{itemize}

Note that the definition above is valid since no vertex in $V(G)$ participates in both a pair and a triple, or in more than one triple, and hence no integer is assigned by any function $\ell_{\diamond}$ more than once. This completes the description of the Triple Selector gadget.

\begin{figure}[t!]\centering
\fbox{\includegraphics[scale=0.8]{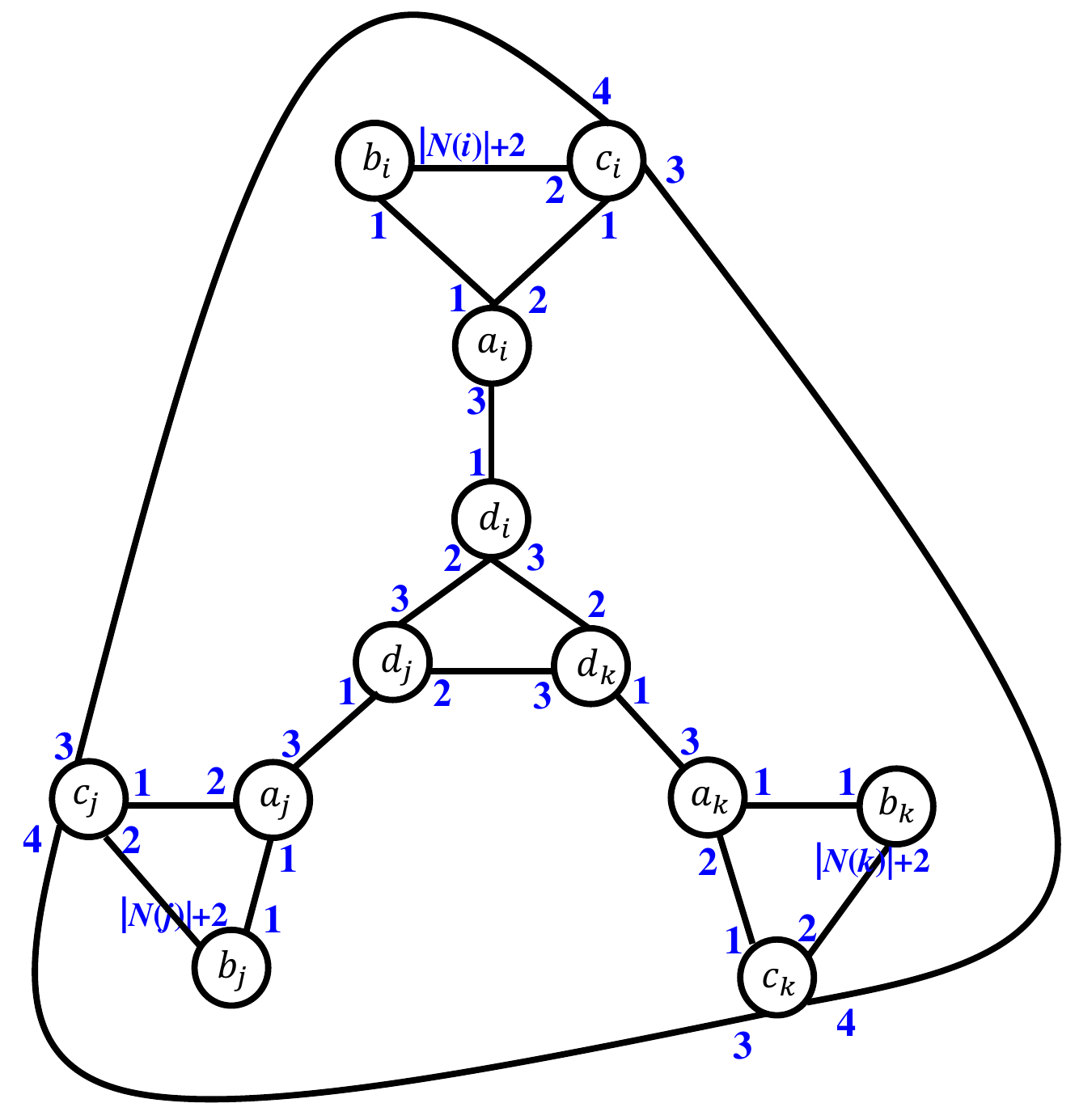}}
\caption{The Triple Selector gadget.}\label{fig:tripleSelector}
\end{figure}

\paragraph{Intuition.} First, we would like to point out that the Triple Selector gadget is symmetric with respect to cyclic shifts. That is, if we replace $j$ by $i$, $k$ by $j$, and $i$ by $k$, then we obtain an isomorphic structure also with respect to preferences. We will use this symmetry when we prove the correctness of our reduction.

To gain some deeper understanding of this gadget, let us recall that in \pvc, {\em exactly} two vertices among $\{i,j,k\}$ must be selected. We already know that the Edge Coverage gadget will ensure that {\em at least} two vertices among $\{i,j,k\}$ are selected (since $\{i,j,k\}$ induces a triangle in $G$ and to cover the edges of a triangle at least two of its vertices must be selected). Hence, we only need to ensure that not {\em all} of the vertices $i,j$ and $k$ are selected. However, if $i,j$ and $k$ are all selected, then $d_i,d_j$ and $d_k$ are all left not matched. Then, the preferences on the triangle on $\{d_i,d_j,d_k\}$ are chosen specifically to ``cause trouble'' in a manner similar to the Pair Selector gadget---again, no matter which edge of this triangle will be picked by the matching, we can replace it by a different edge on this triangle to exhibit a more popular matching. For example, if we pick $\{d_i,d_j\}$, then $d_k$ is left not matched, while $d_j$ prefers $d_k$ over $d_i$. This means that by replacing $\{d_i,d_j\}$ by $\{d_j,d_k\}$, we make both $d_j$ and $d_k$ more satisfied, while only $d_i$ becomes less satisfied (no other vertex in $H$ is affected by the swap).

As in the case of the Pair Selector gadget, the inner triangle (in Fig.~\ref{fig:tripleSelector}) on $\{d_i,d_j,d_k\}$ is not sufficient---the forward direction of the proof fails without the outer triangle on $\{c_i,c_j,c_k\}$. Here, to make the forward direction go through, an additional idea is required. Roughly speaking, we need to have coordination between the triangles (recall that in the previous gadget, some coordination was also noted as a requirement to ensure symmetry, but here deeper coordination is required). Let us elaborate (in an non-formal manner) on the meaning of this coordination here. Specifically, we ``orient'' the inner triangle and the outer triangle in different directions. (Note that symmetry would have been achieved even if we would have oriented them in the same direction.) By this, we mean that while in the inner triangle, $d_i$ prefer $d_j$ over $d_k$, $d_j$ prefers $d_k$ over $d_i$, and $d_k$ prefers $d_i$ over $d_j$, the same does not hold when we rename $d$ to be $c$---here, the direction is reversed, as $c_i$ prefers $c_k$ over $c_j$, $c_j$ prefers $c_i$ over $c_k$, and $c_k$ prefers $c_j$ over $c_i$. This reversal will come in handy when we prove the forward direction, as it will ``block up'' alternating cycles that must not exist by Proposition \ref{prop:char}. Intuitively, the main insight is that if we try to improve the matching we will construct in the proof of the forward direction in a ``clockwise direction'', then we can make two $d$-type vertices more satisfied and only one $d$-type vertex less satisfied, but at the same time, more $c$-type vertices become unsatisfied, and hence we overall do not gain more popularity. In addition, if we try to improve the matching in a ``counter-clockwise direction'', then we can make two $c$-type vertices more satisfied and only one $c$-type vertex less satisfied, but at the same time, more $d$-type vertices become unsatisfied, and hence again we overall do not gain more popularity.

\section{Correctness}\label{sec:correctness}

In this section, we prove the correctness of our reduction. For the sake of clarity, the proof is divided into two subsections, corresponding to the forward and reverse directions. Together with Lemma \ref{lem:partVC}, this proof will conclude the correctness of Theorem \ref{thm:main}.

\subsection{Forward Direction}\label{sec:forward}

Here, we prove that if there exists a solution to the instance $(G,{\cal P},{\cal T})$ of \pvc, then there exists a popular matching in $\red(I)=(H,L=\{\ell_v: v\in V(H)\})$. For this purpose, let us suppose that $U$ is a solution to $(G,{\cal P},{\cal T})$. In what follows, we first construct a matching $M$ in $H$. Then, we will show that the graph $H_M$ (see Definition \ref{def:GM}) satisfies several useful properties, which will eventually lead us to the conclusion that $M$ is popular.

\begin{figure}[t!]\centering
\fbox{\includegraphics[scale=0.8]{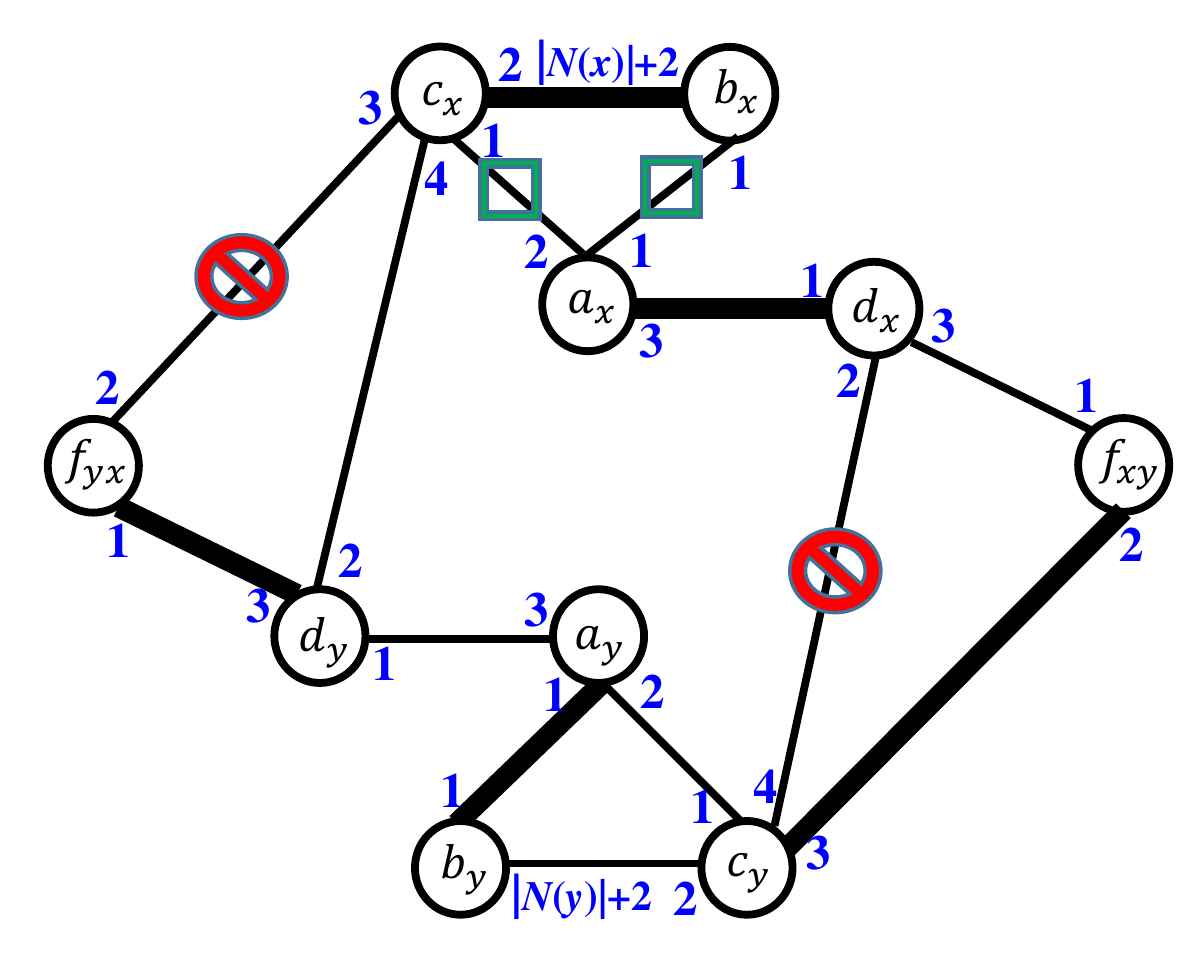}}
\caption{Edges shown in bold are inserted into $M$ (in Section \ref{sec:forward}). Edges labeled -2 by $\lab_M$ are marked by a red no entry sign, and edges labeled +2 by $\lab_M$ are marked by a green square.}\label{fig:pairSelectorMatch}
\end{figure}

\paragraph{Construction of $M$.} The matching $M$ is the union of the following sets.
\begin{itemize}
\item $M_U=\{\{u^e_i,u^e_j\}: \{i,j\}\in E(G)\}$.
\item For every $\{i,j\}\in{\cal P}$ with $i<j$, let $x\in\{i,j\}$ be the vertex not in $U$, and $y\in\{i,j\}$ be the vertex in $U$, and insert the edges $\{a_x,d_x\}$, $\{b_x,c_x\}$, $\{a_y,b_y\}$, $\{f_{xy},c_y\}$ and $\{f_{yx},d_y\}$ into $M_{\cal P}$. (See Fig.~\ref{fig:pairSelectorMatch}.)
\item For every $\{i,j,k\}\in{\cal T}$ with $i<j<k$, let $x\in\{i,j,k\}$ be the vertex not in $U$, and $y,z\in\{i,j,k\}$ be the two vertices in $U$ such that $d_x$ prefers $d_y$ over $d_z$, and insert the edges $\{a_x,d_x\}$, $\{b_x,c_x\}$, $\{a_y,b_y\}$, $\{a_z,b_z\}$, $\{c_y,c_z\}$ and $\{d_y,d_z\}$ into $M_{\cal T}$. (See Fig.~\ref{fig:tripleSelectorMatch}.)
\end{itemize}
Since the sets in ${\cal P}\cup{\cal T}$ are pairwise disjoint, the sets above are well (uniquely) defined. We also remark that the figures do not only capture the case where $i=x$ due to the symmetry of our gadgets (i.e., if $j=x$ or $k=x$ in the case of a triple, we obtain precisely the same figures).

\begin{figure}[t!]\centering
\fbox{\includegraphics[scale=0.8]{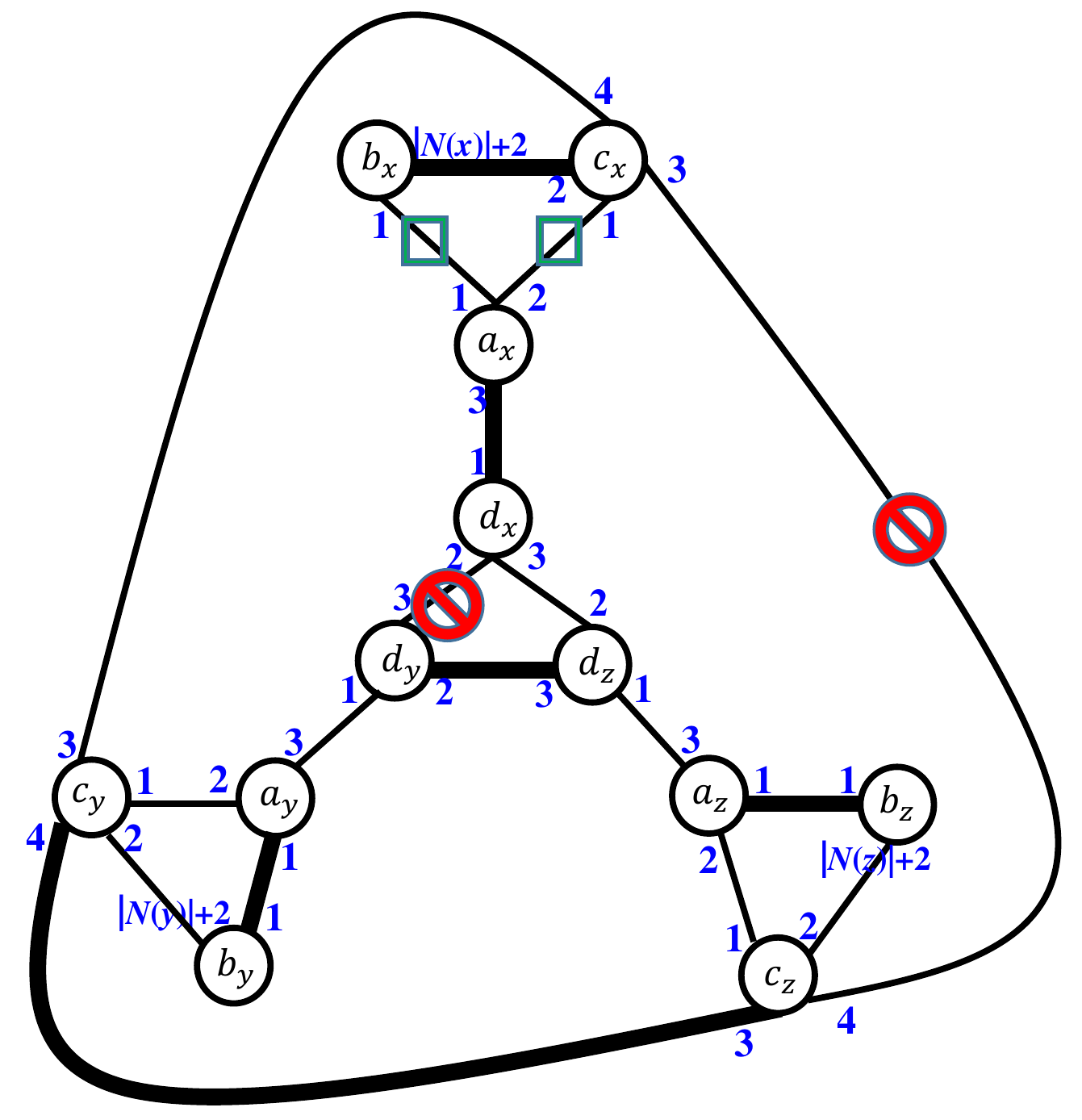}}
\caption{Edges shown in bold are inserted into $M$ (in Section \ref{sec:forward}). Edges labeled -2 by $\lab_M$ are marked by a red no entry sign, and edges labeled +2 by $\lab_M$ are marked by a green square.}\label{fig:tripleSelectorMatch}
\end{figure}

\paragraph{Properties of $H_M$.} Let us start by observing that, since all vertices in $H$ are matched by $M$, the following statement immediately holds.
\begin{observation}\label{obs:noUnmatched}
There is no alternating path in $H_M$ that starts from a vertex not matched by $M$ and contains at least one edge labeled +2 by $\lab_M$.
\end{observation}

We proceed to identify which edges in $H_M$ are labeled +2 by $\lab_M$.

\begin{lemma}\label{lem:2}
The set of edges labeled +2 by $\lab_M$ is $\{\{a_i,b_i\}: i\notin U\}\cup\{\{a_i,c_i\}: i\notin U\}$.
\end{lemma}

\begin{proof}
First, for all $i\notin U$, we have that $\{a_i,d_i\}\in M$ and $\{b_i,c_i\}\in M$. Since $a_i$ prefers both $b_i$ and $c_i$ over $d_i$, and both $b_i$ and $c_i$ prefer $a_i$ over each other, we have that all the edges in $\{\{a_i,b_i\}: i\notin U\}\cup\{\{a_i,c_i\}: i\notin U\}$ are labeled +2 by $\lab_M$. Next, we show that all other edges in $H$ are not labeled +2 by $\lab_M$, which will complete the proof.

Observe that for all $i\in U$, we have that $\{a_i,b_i\}\in M$, and since $p_{a_i}(b_i)=p_{b_i}(a_i)=1$, this means the no edge incident to $a_i$ or $b_i$ can be labeled +2 by $\lab_M$. Similarly, for all $\{i,j\}\in E(G)$, we have that $\{u^e_i,u^e_j\}\in M$, and since $p_{u^e_i}(u^e_j)=p_{u^e_j}(u^e_i)=1$, this means the no edge incident to $u^e_i$ or $u^e_j$ can be labeled +2 by $\lab_M$. Thus, no edge that belongs to an Edge Coverage gadget, excluding the edges in $\{\{a_i,b_i\}: i\notin U\}\cup\{\{a_i,c_i\}: i\notin U\}$, is labeled +2 by $\lab_M$.

Now, consider some pair $\{i,j\}\in{\cal P}$ with $i<j$, and let $x\in\{i,j\}$ be the vertex not in $U$, and $y\in\{i,j\}$ be the vertex in $U$. Then, the edges $\{a_x,d_x\}$, $\{b_x,c_x\}$, $\{a_y,b_y\}$, $\{f_{xy},c_y\}$ and $\{f_{yx},d_y\}$ belong to $M$. However, $c_x$ prefers $b_x$ over both $f_{yx}$ and $d_y$, and $d_x$ prefers $a_x$ over both $f_{xy}$ and $c_y$, which means that none of the edges $\{c_x,f_{yx}\}$, $\{c_x,d_y\}$, $\{d_x,f_{xy}\}$ and $\{d_x,c_y\}$ is labeled +2 by $\lab_M$.

Finally, consider some triple $\{i,j,k\}\in{\cal T}$ with $i<j<k$, and let $x\in\{i,j,k\}$ be the vertex not in $U$, and $y,z\in\{i,j,k\}$ be the two vertices in $U$ such that $d_x$ prefers $d_y$ over $d_z$. Then, the edges $\{a_x,d_x\}$, $\{b_x,c_x\}$, $\{a_y,b_y\}$, $\{a_z,b_z\}$, $\{c_y,c_z\}$ and $\{d_y,d_z\}$ belong to $M$. However, $c_x$ prefers $b_x$ over both $c_y$ and $c_z$, and $d_x$ prefers $a_x$ over both $d_y$ and $d_z$, which means that none of the edges $\{c_x,c_y\}$, $\{c_x,c_z\}$, $\{d_x,d_y\}$ and $\{d_x,d_z\}$ is labeled +2 by $\lab_M$.
\end{proof}

Now, Lemma \ref{lem:2} directly implies the correctness of the following lemma.
\begin{lemma}\label{lem:plusInGadget}
For any $P\in {\cal P}$, the only edges labeled +2 by $\lab_M$ in the Pair Selector gadget associated with $P$ are $\{a_x,b_x\}$ and $\{a_x,c_x\}$ for the unique vertex $x\in P$ that is not in $U$. Similarly, for any $T\in {\cal T}$, the only edges labeled +2 by $\lab_M$ in the Triple Selector gadget associated with $T$ are $\{a_x,b_x\}$ and $\{a_x,c_x\}$ for the unique vertex $x\in T$ that is not in $U$.
\end{lemma}

Having Lemma \ref{lem:plusInGadget} at hand, we are ready to rule out the possibly of having a ``bad'' alternating path that is completely contained inside a Pair Selector gadget or a Triple Selector gadget.

\begin{lemma}\label{lem:noPath}
For any $P\in {\cal P}$, there is no alternating path in $H_M$ that contains at least two edges labeled +2 by $\lab_M$ and which consists only of edges from the Pair Selector gadget associated with $P$.
Similarly, for any $T\in {\cal T}$, there is no alternating path in $H_M$ that contains at least two edges labeled +2 by $\lab_M$ and which consists only of edges from the Triple Selector gadget associated with $T$.
\end{lemma}

\begin{proof}
First, consider some pair $P\in{\cal P}$. By Lemma \ref{lem:plusInGadget}, the only edges labeled +2 by $M$ in the Pair Selector gadget associated with $P$ are $\{a_x,b_x\}$ and $\{a_x,c_x\}$ for the unique vertex $x\in P$ that is not in $U$. However, these two edges are part of a triangle in $H$, and therefore no alternating path can contain both of them together.

Second, consider some triple $T\in{\cal T}$. By Lemma \ref{lem:plusInGadget}, the only edges labeled +2 by $M$ in the Triple Selector gadget associated with $T$ are $\{a_x,b_x\}$ and $\{a_x,c_x\}$ for the unique vertex $x\in T$ that is not in $U$. However, these two edges are again part of a triangle in $H$, and therefore no alternating path can contain both of them together.
\end{proof}

In the following two lemmas, we also rule out the possibly of having a ``bad'' alternating cycle that is completely contained inside a Pair Selector gadget or a Triple Selector gadget.

\begin{figure}[t!]\centering
\fbox{\includegraphics[scale=0.8]{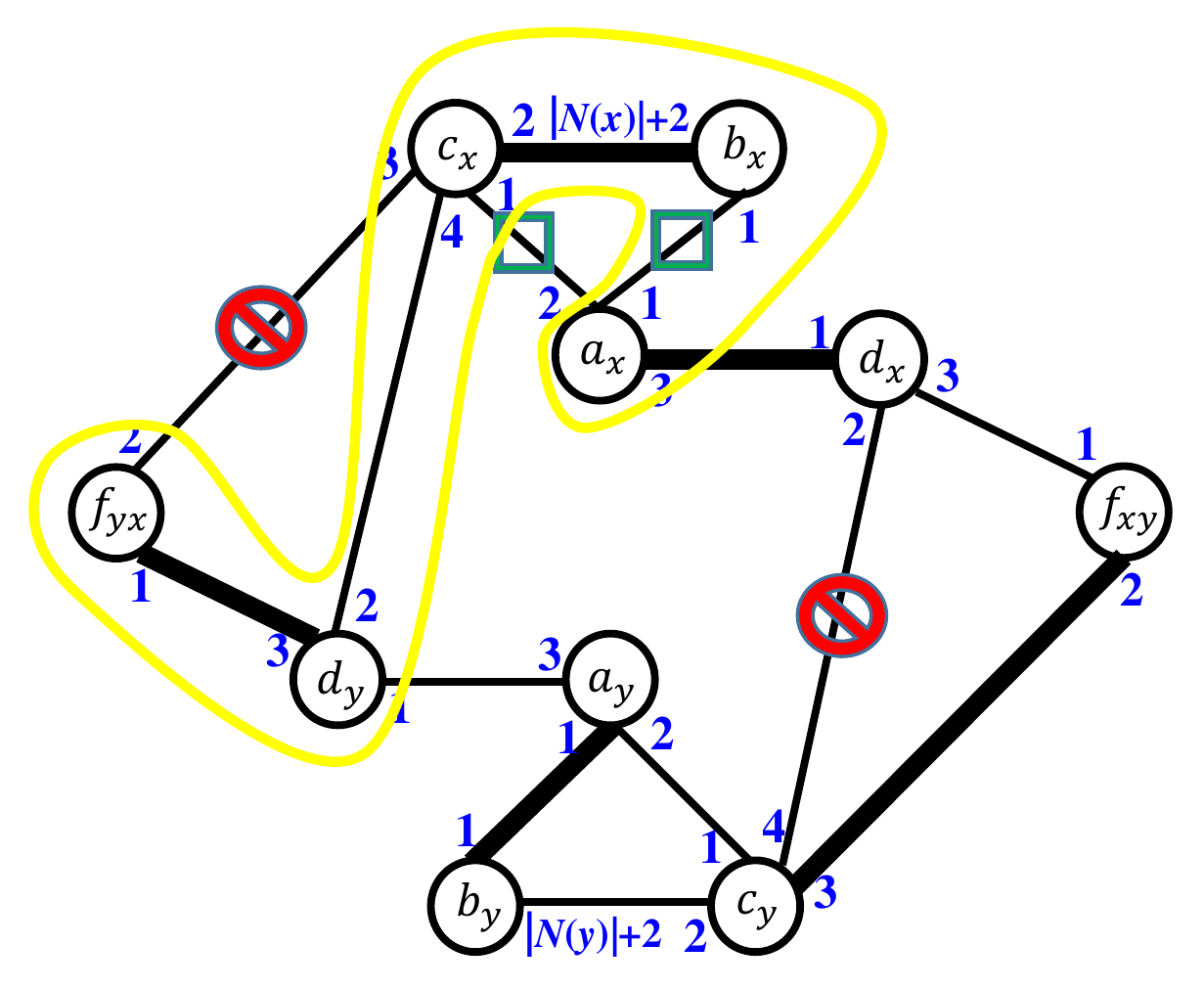}}
\caption{The path constructed in the proof of Lemma \ref{lem:noCyclePair}, highlighted in yellow.}\label{fig:pairSelectorPath}
\end{figure}

\begin{lemma}\label{lem:noCyclePair}
For any $P\in {\cal P}$, there is no alternating cycle in $H_M$ that contains at least one edge labeled +2 by $\lab_M$ and which consists only of edges from the Pair Selector gadget associated with $P$.
\end{lemma}

\begin{proof}
Consider some pair $P=\{i,j\}\in{\cal P}$ with $i<j$, and let $x\in\{i,j\}$ be the vertex not in $U$, and $y\in\{i,j\}$ be the vertex in $U$. Suppose, by way of contradiction, that there exists an alternating cycle $C$ in $H_M$ that contains at least one edge labeled +2 by $\lab_M$ and which consists only of edges from the Pair Selector gadget associated with $P$.

First, suppose that $\{a_x,c_x\}\in E(C)$. Then, since $\{c_x,b_x\}\in M$, we have that $\{c_x,b_x\}\in E(C)$. Since the only neighbor in the gadget of $b_x$ apart from $c_x$ is $a_x$, we have that $\{b_x,a_x\}\in E(C)$. However, we have thus ``closed'' a triangle, which contradicts the choice of $C$ as an alternating cycle.

By Lemma \ref{lem:plusInGadget} and since $C$ contains at least one edge labeled +2 by $\lab_M$, it must hold that $\{a_x,b_x\}\in E(C)$. Then, since $\{c_x,b_x\}\in M$, we have that $\{c_x,b_x\}-\{b_x,a_x\}$ is a subpath of $C$. Now, note that $c_x$ prefers $b_x$ over its two other neighbors in the gadget, and $f_{yx}$ prefers $d_y$ over $c_x$. Therefore, $\{c_x,f_{yx}\}$ is labeled -2 by $\lab_M$, and hence it does not exist in $H_M$. Thus, we also have that $\{d_y,c_x\}\in E(C)$, and since $\{d_y,f_{yx}\}\in M$, we have that $\{f_{yx},d_y\}-\{d_y,c_x\}-\{c_x,b_x\}-\{b_x,a_x\}$ is a subpath of $C$ (see Fig.~\ref{fig:pairSelectorPath}). However, $f_{yx}$ has no neighbor in $H_M$ apart from $d_y$, and therefore we have reached a contradiction to the choice of $C$ as an alternating cycle.
\end{proof}

\begin{lemma}\label{lem:noCycleTriple}
For any $T\in {\cal T}$, there is no alternating cycle in $H_M$ that contains at least one edge labeled +2 by $\lab_M$ and which consists only of edges from the Triple Selector gadget associated with $T$.
\end{lemma}

\begin{proof}
Consider some triple $T=\{i,j,k\}\in{\cal T}$ with $i<j<k$, and let $x\in\{i,j,k\}$ be the vertex not in $U$, and $y,z\in\{i,j,k\}$ be the two vertices in $U$ such that $d_x$ prefers $d_y$ over $d_z$. Suppose, by way of contradiction, that there exists an alternating cycle $C$ in $H_M$ that contains at least one edge labeled +2 by $\lab_M$ and which consists only of edges from the Triple Selector gadget associated with $T$.

First, suppose that $\{a_x,c_x\}\in E(C)$. Then, since $\{c_x,b_x\}\in M$, we have that $\{c_x,b_x\}\in E(C)$. Since the only neighbor in the gadget of $b_x$ apart from $c_x$ is $a_x$, we have that $\{b_x,a_x\}\in E(C)$. However, we have thus ``closed'' a triangle, which contradicts the choice of $C$ as an alternating cycle.

By Lemma \ref{lem:plusInGadget} and since $C$ contains at least one edge labeled +2 by $\lab_M$, it must hold that $\{a_x,b_x\}\in E(C)$. Then, since $\{c_x,b_x\}\in M$ and $\{a_x,d_x\}\in M$, we have that $\{c_x,b_x\}-\{b_x,a_x\}-\{a_x,d_x\}$ is a subpath of $C$. Observe that $c_x$ prefers $b_x$ over $c_z$, and $c_z$ prefers $c_y$ over $c_x$. Moreover, $d_x$ prefers $a_x$ over $d_y$, and $d_y$ prefers $d_z$ over $d_x$. Therefore, both $\{c_x,c_z\}$ and $\{d_x,d_y\}$ are labeled -2 by $\lab_M$, which means that these two edges do not exist in $H_M$. Since the only neighbor of $c_x$ in the gadget except for $a_x,b_x$ and $c_z$ is $c_y$, and since the only neighbor of $d_x$ in the gadget except for $a_x$ and $d_y$ is $d_z$, we have that $\{c_y,c_x\},\{d_x,d_z\}\in E(C)$. Since $\{c_z,c_y\},\{d_z,d_y\}\in M$, this means that $\{c_z,c_y\}-\{c_y,c_x\}-\{c_x,b_x\}-\{b_x,a_x\}-\{a_x,d_x\}-\{d_x,d_z\}-\{d_z,d_y\}$ is a subpath of $C$. Now, since the only neighbor of $d_y$ in the gadget except for $d_x$ and $d_z$ is $a_y$, we have that $\{d_y,a_y\}\in E(C)$. Because $\{a_y,b_y\}\in M$, and since the only neighbor of $b_y$ in this gadget except for $a_y$ is $c_y$, this means that $\{c_z,c_y\}-\{c_y,c_x\}-\{c_x,b_x\}-\{b_x,a_x\}-\{a_x,d_x\}-\{d_x,d_z\}-\{d_z,d_y\}-\{d_y,a_y\}-\{a_y,b_y\}-\{b_y,c_y\}$ is a subpath of $C$ (see Fig.~\ref{fig:tripleSelectorPath}). However, $c_y$ has three different neighbors on this path, which contradicts the choice of $C$ as an alternating cycle.
\end{proof}

\begin{figure}[t!]\centering
\fbox{\includegraphics[scale=0.8]{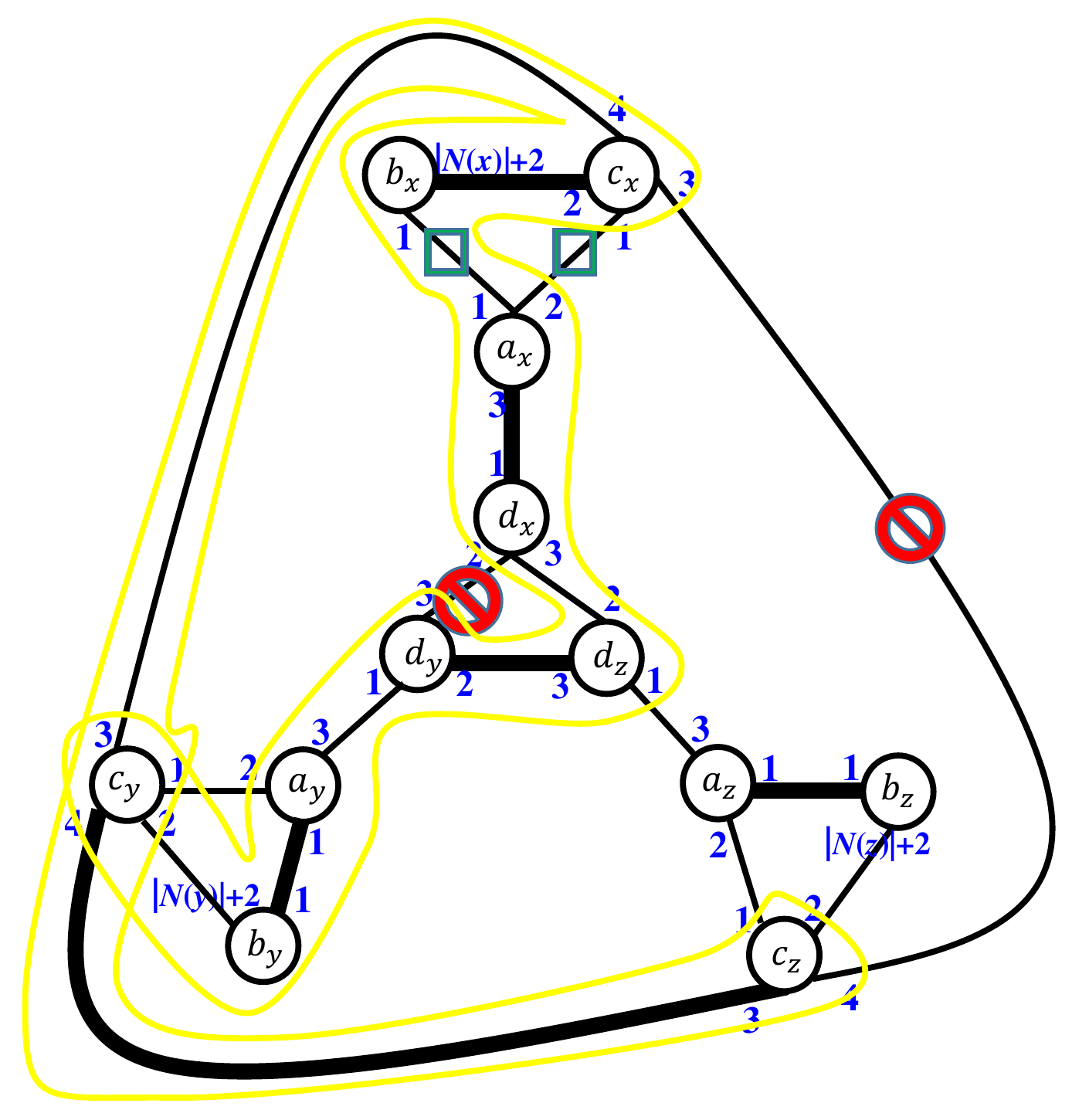}}
\caption{The path constructed in the proof of Lemma \ref{lem:noCycleTriple}, highlighted in yellow.}\label{fig:tripleSelectorPath}
\end{figure}

Next, in the following two lemmas, we argue that a shortest ``bad'' alternating path, as well as a ``bad'' alternating cycle, cannot contain any vertex of the form $u^e_i$. The proof of the second lemma is essentially a simplified version of the proof of the first one, but we present the full details for the sake of clarity.

\begin{lemma}\label{lem:noCrossPath}
Let $S$ be a shortest alternating path in $H_M$ that contains at least two edges labeled +2 by $\lab_M$. Then, $S$ does not contain the vertex $u^e_i$ for any $e\in E(G)$ and $i\in e$.
\end{lemma}

\begin{proof}
Suppose, by way of contradiction, that $S$ contains the vertex $u^e_i$ for some $e\in E(G)$ and $i\in e$. Then, since $\{u^e_i,u^e_j\}\in M$ for the vertex $j\in V(G)$ such that $e=\{i,j\}$, we have that $\{u^e_i,u^e_j\}\in E(S)$. Since $U$ is a vertex cover in $G$, at least one vertex in $\{i,j\}$ belongs to $U$, and let us suppose without loss of generality that this vertex is $i$. Then, $\{a_i,b_i\}\in M$, which means that $\lab_M(\{b_i,u^e_i\})=-2$. Thus, $\{b_i,u^e_i\}\notin E(H_M)$. Since the only neighbors of $u^e_i$ in $H$ are $b_i$ and $u^e_j$,  the only neighbors of $u^e_j$ in $H$ are $b_j$ and $u^e_i$, this implies that $u^i_e$ is an endpoint of $S$ and $\{u^e_j,b_j\}\in E(S)$. Since $u^e_j$ prefers $u^e_i$ over $b_j$, we have that $\lab_M(\{b_j,u^e_j\})\neq +2$. Thus, by removing $\{u^e_i,u^e_j\}$ and $\{b_j,u^e_j\}$ from $S$, we obtain yet another alternating path in $H_M$ that contains at least two edges labeled +2 by $\lab_M$. This contradicts the choice of $S$ as the shortest alternating path in $H_M$ with this property.
\end{proof}

\begin{lemma}\label{lem:noCrossCycle}
There is no alternating cycle in $H_M$ that contains the vertex $u^e_i$ for any $e\in E(G)$ and $i\in e$.
\end{lemma}

\begin{proof}
Suppose, by way of contradiction, that there exists an alternating cycle $C$ in $H_M$ that contains the vertex $u^e_i$ for some $e\in E(G)$ and $i\in e$. Then, since $\{u^e_i,u^e_j\}\in M$ for the vertex $j\in V(G)$ such that $e=\{i,j\}$, we have that $\{u^e_i,u^e_j\}\in E(C)$. Since $U$ is a vertex cover in $G$, at least one vertex in $\{i,j\}$ belongs to $U$, and let us suppose without loss of generality that this vertex is $i$. Then, $\{a_i,b_i\}\in M$, which means that $\lab_M(\{b_i,u^e_i\})=-2$. Thus, $\{b_i,u^e_i\}\notin E(H_M)$. Since the only neighbors of $u^e_i$ in $H$ are $b_i$ and $u^e_j$, this implies that $C$ cannot be a cycle, and hence we have reached a contradiction.
\end{proof}

\paragraph{Conclusion of the forward direction.} First, by Observation \ref{obs:noUnmatched}, there is no alternating path in $H_M$ that starts from a vertex not matched by $M$ and contains at least one edge labeled +2 by $\lab_M$. Now, by Lemmas \ref{lem:noCrossPath} and \ref{lem:noCrossCycle}, if there exists an alternating cycle in $H_M$ that contains at least one edge labeled +2 by $\lab_M$, or an alternating path in $H_M$ that contains at least two edges labeled +2 by $\lab_M$, then there also exists such a cycle or path that does not contain any vertex in $\{u^e_i: e\in E(G), i\in e\}$. However, if we remove the vertices in $\{u^e_i: e\in E(G), i\in e\}$ from $H$, then the remaining connected components are precisely the Pair Selector and Triple Selector gadgets. By Lemmas \ref{lem:noPath} and \ref{lem:noCyclePair}, there exists no alternating cycle in $H_M$ that contains at least one edge labeled +2 by $\lab_M$, as well as no  alternating path in $H_M$ that contains at least two edges labeled +2 by $\lab_M$, which consists only of edges of a Pair Selector gadget. Moreover, by Lemmas \ref{lem:noPath} and \ref{lem:noCycleTriple}, the same claim holds also with respect to a Triple Selector gadget. Thus, by Proposition \ref{prop:char}, we conclude that $M$ is popular.

\subsection{Reverse Direction}

Here, we prove that if there exists a popular matching in $\red(I)=(H,L=\{\ell_v: v\in V(H)\})$, then there exists a solution to the instance $(G,{\cal P},{\cal T})$ of \pvc. For this purpose, let us suppose that $M$ is a popular matching in $(H,L=\{\ell_v: v\in V(H)\})$. In what follows, we first construct a subset $U\subseteq V(G)$. Then, we will show that $U$ is a vertex cover of $G$. Afterwards, we will show that for every $P\in {\cal P}$, it holds that $|U\cap P|=1$. Lastly, we will show that for every $T\in {\cal T}$, it holds that $|U\cap T|=2$, which will conclude the proof. 

Before implementing this plan, let us give a folklore observation (that is true for any graph and preference lists) that will be used in all proofs ahead.

\begin{observation}\label{obs:maximal}
Let $J$ be an instance of \popular. Every popular matching in $J$ is a maximal matching.
\end{observation}

\begin{proof}
Suppose, by way of contradiction, that there exists a popular matching $\widehat{M}$ in $J$ that is not maximal. Then, there exists an edge $\{x,y\}$ that is present in the graph in $J$, and with both endpoints not matched by $\widehat{M}$. However, by adding $\{x,y\}$ to $\widehat{M}$, we obtain a matching more popular than $\widehat{M}$, and thus reach a contradiction.
\end{proof}

\paragraph{Construction of $U$.} We simply define $U:=\{i\in V(G): \{a_i,b_i\}\in M\}$.

\paragraph{Proof that $U$ is a vertex cover.} The proof the $U$ is a vertex cover is the same as a proof given by Kavitha \cite{DBLP:journals/corr/abs-1802-07440}. However, for the sake of completeness, and also to verify that although our construction has other components, that same proof still goes through, we will present the details in Appendix \ref{app:vc}. Here, we state the claim we need in the following lemma.

\begin{lemma}\label{lem:vc}
The set $U$ is a vertex cover of $G$.
\end{lemma}

\paragraph{Proof that $U$ is a solution.} Since we have already established that $U$ is a vertex cover, the proof that $U$ is a solution will follow from the correctness of the two following lemmas.

\begin{lemma}
For every $P\in {\cal P}$, it holds that $|U\cap P|=1$.
\end{lemma}

\begin{proof}
Let us consider some arbitrary pair $P=\{i,j\}\in {\cal P}$. By Lemma \ref{lem:vc}, and because a pair is also an edge in $G$, we have that $|U\cap P|\geq 1$. Thus, to prove the lemma, it suffices to show that it is not possible to have $|U\cap P|=2$. To this end, suppose by way of contradiction that $|U\cap P|=2$. By the definition of $U$, both $\{a_i,b_i\}\in M$ and $\{a_j,b_j\}\in M$. Note that the only neighbors of $c_i$ besides $a_i$ and $b_i$ are $f_{ji}$ and $d_j$, the only neighbors of $f_{ji}$ are $c_i$ and $d_j$, and the only neighbors of $d_j$ besides $a_j$ are $c_i$ and $f_{ij}$. Thus, by Observation \ref{obs:maximal}, $M$ must contain exactly one of the edges $\{c_i,d_j\}$, $\{d_j,f_{ji}\}$ and $\{f_{ji},c_i\}$. If $\{c_i,d_j\}\in M$, then by replacing this edge by $\{f_{ji},c_i\}$, we obtain a more popular matching (both $c_i$ and $f_{ji}$ vote in favor of the replacement, while only $d_j$ votes against it). If $\{d_j,f_{ji}\}\in M$, then by replacing this edge by $\{c_i,d_j\}$, we obtain a more popular matching (both $c_i$ and $d_j$ vote in favor of the replacement, while only $f_{ji}$ votes against it). If $\{f_{ji},c_i\}\in M$, then by replacing this edge by $\{d_j,f_{ji}\}$, we obtain a more popular matching (both $d_j$ and $f_{ji}$ vote in favor of the replacement, while only $c_i$ votes against it). Since every case led to a contradiction, the proof is complete.
\end{proof}

\begin{lemma}
For every $T\in {\cal T}$, it holds that $|U\cap T|=2$.
\end{lemma}

\begin{proof}
Let us consider some arbitrary triple $T=\{i,j,k\}\in {\cal T}$. By Lemma \ref{lem:vc}, and because a triple is also a triangle in $G$, we have that $|U\cap T|\geq 2$. Thus, to prove the lemma, it suffices to show that it is not possible to have $|U\cap T|=3$. To this end, suppose by way of contradiction that $|U\cap T|=3$. By the definition of $U$, all the three edges $\{a_i,b_i\}$, $\{a_j,b_j\}$ and $\{a_k,b_k\}$ belong to $M$. Note that the only neighbors of $d_i$ besides $a_i$ are $d_j$ and $d_k$, the only neighbors of $d_j$ besides $a_j$ are $d_i$ and $d_k$, and the only neighbors of $d_k$ besides $a_k$ are $d_i$ and $d_j$. Thus, by Observation \ref{obs:maximal}, $M$ must contain exactly one of the edges $\{d_i,d_j\}$, $\{d_j,d_k\}$ and $\{d_k,d_i\}$. If $\{d_i,d_j\}\in M$, then by replacing this edge by $\{d_j,d_k\}$, we obtain a more popular matching (both $d_j$ and $d_k$ vote in favor of the replacement, while only $d_i$ votes against it). If $\{d_j,d_k\}\in M$, then by replacing this edge by $\{d_k,d_i\}$, we obtain a more popular matching (both $d_i$ and $d_k$ vote in favor of the replacement, while only $d_j$ votes against it). If $\{d_k,d_i\}\in M$, then by replacing this edge by $\{d_i,d_j\}$, we obtain a more popular matching (both $d_i$ and $d_j$ vote in favor of the replacement, while only $d_k$ votes against it). Since every case led to a contradiction, the proof is complete.
\end{proof}

\bibliographystyle{siam}
\bibliography{references}

\appendix
\section{Proof of Lemma \ref{lem:partVC}}\label{app:partVC}

In an instance of {\sc 3-SAT}, we are given a set of variables $X$, and a formula encoded as a collection of clauses ${\cal C}$. Each clause $C\in{\cal C}$ is a set of exactly three literals, where each literal is either a variable $x\in X$ or the negation of a variable $x\in X$ that is denoted by $\overline{x}$. A truth assignment $\alpha: X\rightarrow\{\mathsf{T},\mathsf{F}\}$ satisfies a literal $\ell$ if either it is positive and assigned truth, or negative and assigned false. Now, $\alpha$ satisfies a clause if it satisfies at least one of its literals, and it satisfies ${\cal C}$ if it satisfy every clause in $\cal C$. The objective is to decide whether there exists a truth assignment that satisfies $\cal C$.

The {\sc 3-SAT} problem is \NPH, and below we give the well-known classic reduction from {\sc 3-SAT} to {\sc Vertex Cover} that shows that {\sc Vertex Cover} is \NPH. This result is summarized in Proposition \ref{prop:normalVC}. Afterwards, we argue how the instance outputted by the reduction can be viewed as an instance of \pvc, and thus conclude the proof of Lemma \ref{lem:partVC}.

\paragraph{Reduction from {\sc 3-SAT} to \pvc.} Let $I=(X,{\cal C})$. Then, we construct an instance $\red(I)=(G,k)$ of {\sc Vertex Cover} as follows. First, define $k=|X|+2|{\cal C}|$. Now, for every $x\in X$, add two new vertices $v_x$ and $v_{\overline{x}}$, along with the edge $\{v_x,v_{\overline{x}}\}$, to $G$. For every clause $C=\{p,q,r\}\in{\cal C}$, add three new vertices $u^C_p,u^C_q$ and $u^C_r$, along with the edges $\{u^C_p,u^C_q\}$, $\{u^C_q,u^C_r\}$ and $\{u^C_r,u^C_p\}$ to $G$. Finally, for every $C\in{\cal C}$ and $\ell\in C$, add the edge $\{u^C_\ell,v_\ell\}$ to $G$. It is easy to verify that the following result holds (see, e.g., \cite{sipser2006}).

\begin{proposition}[\cite{sipser2006}]\label{prop:normalVC}
Let $I=(X,{\cal C})$ be an instance of {\sc 3-SAT}. Then, $I$ has a satisfying assignment if and only if $G$ has a vertex cover of size at most $k$ where $(G,k)=\red(I)$.
\end{proposition}

\paragraph{The viewpoint of \pvc.} Given the output instance $(G,k)$ of the reduction, we define ${\cal P}=\{\{v_x,v_{\overline{x}}: x\in X\}$ and ${\cal T}=\{\{u^C_p,u^C_q,u^C_r\}: C=\{p,q,r\}\in{\cal C}\}$. Clearly, the sets in ${\cal P}\cup{\cal T}$ are pairwise disjoint, every set in ${\cal P}$ is an edge in $G$, and every set in ${\cal T}$ induces a triangle in $G$. Moreover, every vertex cover of $G$ must select at least one vertex of each edge in $E(G)$, and at least two vertices of every triangle in $G$. Since $|{\cal P}|=|X|$ and $|{\cal T}|=2|{\cal C}|$, this means that $G$ has a vertex cover of size at most $k$ if and only if $G$ has a vertex cover of size exactly $k$, and the latter statement holds if and only if $(G,{\cal P},{\cal T})$ has a solution. By Proposition \ref{prop:normalVC}, this concludes the proof of Lemma \ref{lem:partVC}.

\section{Proof of Lemma \ref{lem:vc}}\label{app:vc}

Towards the proof of Lemma \ref{lem:vc}, we first prove the following claim.

\begin{claim}\label{claim:configuration}
For every $i\in V(G)$, we have that
\begin{itemize}
\item $\{a_i,b_i\}\in M$, or
\item both $\{a_i,d_i\}\in M$ and $\{b_i,c_i\}\in M$.
\end{itemize}
\end{claim}

\begin{proof}
Consider some $i\in V(G)$. First, note that $a_i$ is the top preference of every vertex among its neighbors. Therefore, if $a_i$ is not matched by $M$, then we can match it to any of its neighbors (so if that neighbor was matched, its former matching partner is now unmatched), and get two votes in favor of the change and at most one vote against it, which means that $M$ is not popular. Therefore, $a_i$ must be matched by $M$.

Let us now rule out the possibility that $a_i$ is matched to $c_i$ by $M$. If $a_i$ is matched to $c_i$, then $b_i$ must be matched to $u^e_i$ for some edge $e$ incident to $i$, as otherwise by removing $\{a_i,c_i\}$ from $M$ and adding $\{a_i,b_i\}$ to it instead, we obtain a more popular matching. Denote $e=\{i,j\}$. Then, by removing $\{a_i,c_i\}$, $\{b_i,u^e_i\}$ and $\{u^e_j,b_j\}$ (if $\{u^e_j,b_j\}\in M$), and adding $\{a_i,b_i\}$ and $\{u^e_i,u^e_j\}$, we get four votes in favor of the change (from $a_i, b_i, u^e_i$ and $u^e_j$) and at most two votes against it (from $c_i$ and possibly $b_j$), which contradicts the popularity of $M$.

It remains to show that if $a_i$ is matched to $d_i$, then $b_i$ is matched to $c_i$. To this end, suppose that $a_i$ is matched to $d_i$. If $b_i$ is unmatched, then by removing $\{a_i,d_i\}$ and adding $\{b_i,a_i\}$, we obtain a more popular matching. Thus, if $b_i$ is not matched to $c_i$, then it must be matched to $u^e_i$ for some edge $e=\{i,j\}$ incident to $i$. In this case, where $b_i$ is matched to $u^e_i$, by removing $\{a_i,d_i\}$, $\{b_i,u^e_i\}$ and $\{u^e_j,b_j\}$ (if $\{u^e_j,b_j\}\in M$), and adding $\{a_i,b_i\}$ and $\{u^e_i,u^e_j\}$, we get four votes in favor of the change (from $a_i,b_i,u^e_i$ and $u^e_j$) and at most two votes against it (from $d_i$ and possibly $b_j$), which contradicts the popularity of $M$.
\end{proof}

We now proceed with the proof of the lemma. To this end, let $\{i,j\}\in E(G)$ be an arbitrarily chosen edge. To prove that this edge is covered by $U$, we need to show that at least one edge among $\{a_i,b_i\}$ and $\{a_j,b_j\}$ is in $M$. Suppose, by way of contradiction, that this statement is false. Then, by Claim \ref{claim:configuration}, it holds that $\{a_i,d_i\},\{b_i,c_i\},\{a_j,d_j\},\{b_j,c_j\}\in M$. In this case, $\{u^e_i,u^e_j\}$ must be in $M$ (else they are not matched, which contradicts Observation \ref{obs:maximal}). Then, we remove $\{a_i,d_i\},\{b_i,c_i\},\{u^e_i,u^e_j\},\{a_j,d_j\}$ and $\{b_j,c_j\}$ from $M$, and add $\{a_i,c_i\},\{b_i,u^e_i\},\{a_j,c_j\}$ and $\{b_j,u^e_j\}$ to $M$. Then, we gain six votes in favor of the replacement (from $a_i,a_j,b_i,b_j,c_i$ and $c_j$) and only four votes against it (from $d_i,d_j,u^e_i$ and $u^e_j$), which contradicts the popularity of $M$. This completes the proof of the lemma.

\end{document}